\documentclass[11pt]{amsart}
\usepackage[utf8]{inputenc}
 \usepackage{amsaddr}
 
\usepackage{tikz}
\usepackage{multirow}
\usepackage{amsmath,amssymb,amsthm}
\usepackage[linesnumbered,ruled,vlined]{algorithm2e}
\SetKwInput{KwInput}{Input}                
\SetKwInput{KwOutput}{Output}              

\theoremstyle{plain}
\newtheorem{proposition}{Proposition}

\newtheorem{lemma}{Lemma}
\theoremstyle{definition}
\newtheorem{definition}{Definition}
\theoremstyle{remark}

\begin{document}

\title[Weighted kappa with Markov moves]{Analysis of the weighted kappa and its maximum with Markov moves}
\author{Fabio Rapallo}
\address{Department of Economics \\ University of Genova, Italy}
\email{fabio.rapallo@unige.it}


\maketitle

\begin{abstract}
In this paper the notion of Markov move from Algebraic Statistics is used to analyze the weighted kappa indices in rater agreement problems. In particular, the problem of the maximum kappa and its dependence on the choice of the weighting schemes are discussed. The Markov moves are also used in a simulated annealing algorithm to actually find the configuration of maximum agreement.
\medskip

\medskip

{\bf AMS Math Subject Classification:} 62H17; 62H20; 62P10; 62P15 
\end{abstract}

\section{Introduction}

The analysis of rater agreement is currently one among the most active and relevant research areas in categorical data analysis. Despite the large number of papers devoted to this problem, still there are open questions, from the point of view of both theory and applications. Even in the simplest case where two or more observers rate a common set of $n$ objects on the same rating scale, there are in literature several indices to summarize the agreement, each of them with its own paradoxes, counterexamples, and unexpected behaviors. Indeed, the large spectrum of possible indices is the symptom of the difficulties in the interpretation of the results. Also in the context of statistical modeling, several recipes have been proposed in order to account for the rater agreement, but even in this framework no definite answers are available. For a general survey on the main techniques for rater agreement analysis, the reader can refer to \cite{fleiss|etal:03}, \cite{voneye|mun:04} or \cite{shoukri:10}.

The most popular measures of agreement, at least in the two-rater case, are the Cohen's $\kappa$ and the weighted Cohen's $\kappa_w$. First introduced in \cite{cohen:60} and \cite{cohen:68} respectively, such two indices have been analyzed, criticized, generalized, in order to adapt to the multi-rater case, to incomplete rating schemes, and so on. For instance, in the multi-rater case the most popular extension of the Cohen's $\kappa$ is the Conger's $\kappa_C$ introduced in \cite{conger:80}, where the pairwise agreement in all possible two-way marginal tables is considered, see the discussion and the examples in \cite{vanbelle:18}. In all cases, the rationale behind such indices is the measurement of the rater agreement beyond chance, in the sense that under complete independence of the raters the value of the indices should be zero. In this paper we restrict our attention to the weighted Cohen's $\kappa_w$ and its extensions to the multi-rater case, paying special attention to the connections between the choice of the weighting scheme and the maximum attainable value of $\kappa_w$. 

The unweighted version of the Cohen's $\kappa$ only distinguishes between agreement cells and disagreement cells, and thus it is used in case of ratings on a nominal scale. When the rating scale is ordinal, or in general when there are some disagreements to be considered more serious than others, then the weighted $\kappa_w$ should be preferred. In this case, the choice of the weights is a delicate issue, and it is known that different weighting schemes lead to quite different results. In addition, the main weighting schemes (i.e., linear or quadratic) have been studied extensively. Each of them has its own theoretical properties. We will provide some more details and references in the next section, after the definitions. 

As another issue of the kappa-type statistics, it is known that its interpretation is not straightforward since the theoretical maximum is $1$, as a consequence of the normalization, but such theoretical maximum is not attainable when the marginal distributions are not homogeneous. Some attempts has been made in the direction of finding the maximum value of the kappa statistics. For instance, a procedure has been introduced in \cite{umesh|etal:89}, where the maximum agreement is found by fixing the observed agreement and by varying the marginal distributions.

In a recent paper by Kv{\aa}lseth \cite{kvalseth:18}, the dependence of the weighted $\kappa_w$ on the choice of the weights is highlighted, and the relevance of the interpretation of the $\kappa_w$ values as functions of the weighting schemes is discussed extensively. The author motivates its study on the properties of the weights claiming that, without a clear understanding of the connections between the weights and the $\kappa_w$, the weighted $kappa_w$ itself is not a satisfactory index to describe the agreement in an ordinal context. 

The aim of this paper is twofold. First, we give insights in the interpretation of the weights in the framework of weighted kappa statistics by means of the  Markov moves from Algebraic Statistics. In particular, we show how the properties of the weights affect the configuration of maximum agreement. When fixing the marginal distributions, the computation of the maximum attainable kappa is relatively easy in the case of the unweighted $\kappa$ in two-rater setting, see e.g. \cite{sim|wright:05}, the problem is less simple in the weighted case or in the multi-rater setting. The use of Algebraic Statistics for rater agreement analysis has been considered in other works, but mainly for computational purposes. For instance, the use of Markov bases to make exact tests in this framework can be found in \cite{rapallo:03} and \cite{rapallo:05}.

As a second issue, we move to computations and, exploiting again the Markov moves, we introduce a simple simulated annealing algorithm to find the maximum agreement. In particular, we assume the marginal distributions as fixed and we consider all multivariate tables with fixed one-way margins. Also in this context the relevance of the choice of the weights is discussed. The computation of the maximum kappa is implemented in the R package {\tt rel}, \cite{rel:20} and, for the two-rater setting, in some online calculators. 
However, such resources give only an approximate result. Moreover, our proposed algorithm can be applied with a general weighting scheme, not limited to linear or quadratic. 

The paper is organized as follows. In Section \ref{sect:basics} we recall the notation and the basic definitions about the Cohen's $\kappa$, the weighted Cohen's $\kappa_w$, and the Conger's $\kappa_C$ and $\kappa_{C,w}$ for the multi-rater case. In Section \ref{sect:markov} we compute the Markov bases for the rater agreement problem in the two-rater and in the multi-rater cases. Such Markov bases are used in Section \ref{sect:interpret} to state some results on the structure of the configuration of maximum agreement in connection with the (metric) properties of the weighting schemes. Section \ref{sect:algo} is devoted to the illustration of a Simulated Annealing algorithm to actually find the configuration of maximum agreement, while in Section \ref{sect:simst} the results of a simulation study are presented and discussed. Finally, \ref{sect:final} contains some concluding remarks and pointers to future directions.

\section{Notation and basic recalls} \label{sect:basics}

In this section we briefly review the basic definitions about the kappa-type indices of agreement which will be used in the paper. We first focus on the two-rater setting.

Let us consider the ratings of the two raters as a pair of random variables $X$ and $Y$ on the set $\{1, \ldots, k\}$, or more generally on a finite ground set $\{x_1, \ldots, x_k\}$. Let us denote with $p_{ij}$ the probability of the cell $(i,j)$, and with $p_{i+}$ ($i=1, \ldots, k$) and $p_{+j}$ ($j=1, \ldots, k$) the marginal distributions of $X$ and $Y$, respectively. The Cohen's $\kappa$ is defined as:
\begin{equation} \label{kappa}
\kappa = \frac {\sum_{i=1}^k p_{ii} - \sum_{i=1}^k p_{i+}p_{+i}} {1 - \sum_{i=1}^k p_{i+}p_{+i}} = 1 - \frac {\sum_{(i,j) \in D} p_{ij}}{\sum_{(i,j) \in D} p_{i+}p_{+j}} \, ,
\end{equation}
where $D=\{(i,j) \ : \ i \ne j\}$ is the set of the disagreement cells. 

Given a matrix of weights of agreement $W=(w_{ij})$ with $0 \le w_{ij} < 1$ for all $i,j$ with $i \ne j$, and $w_{ii}=1$ for all $i$, the weighted kappa is:
\begin{equation} \label{wkappa}
\kappa_w = \frac {\sum_{i,j=1}^k w_{ij}p_{ij} - \sum_{i,j=1}^k w_{ij}p_{i+}p_{+j}} {1 - \sum_{i,j=1}^k w_{ij}p_{i+}p_{+j}} = 1 - \frac {\sum_{(i,j) \in D} u_{ij} p_{ij}}{\sum_{(i,j) \in D} u_{i,j}p_{i+}p_{+j}} \, ,
\end{equation}
where in the second expression $u_{ij}=1-w_{ij}$. Although not strictly necessary for the theory of rater agreement, we suppose that the matrices $W$ and $U=(u_{ij})$ are symmetric, because some of our results are based on the properties of the metric functions, where symmetry is one the axioms. In the previous formulas, the $u_{ij}$ are weights of disagreement, and it is easily seen that $u_{ij}=0$ on the main diagonal and $0<u_{ij}\le 1$ for $i \ne j$. 

When a sample is available, the indices $\kappa$ and $\kappa_w$ are estimated by replacing in Equations \eqref{kappa} and \eqref{wkappa} the theoretical probabilities with the corresponding sample proportions. On a sample of size $N$, we denote with $n_{ij}$ the count of the cell $(i,j)$ and therefore sample proportion is $\hat p_{ij}=n_{ij}/N$. 

Among the most commonly used weighting schemes there are:
\begin{itemize}
\item[(a)] the quadratic weights (see \cite{fleiss|cohen:73}):
\begin{equation}\label{weight-squared}
u_{ij} = \frac {(i-j)^2} {(k-1)^2}
\end{equation}
\item[(b)] the linear weights (see \cite{cicchetti|allison:71}):
\begin{equation}\label{weight-linear}
u_{ij} = \frac {|i-j|}{k-1}
\end{equation}
\end{itemize}
Moreover, the unweighted $\kappa$ in Eq.~\eqref{kappa} can be considered as a special case of the weighted $\kappa_w$ by setting
\begin{equation}\label{weight-dirac}
u_{ij} = \left\{\begin{array}{ll}0 \ & \ \mbox{ for } i=j \\
1 \ & \ \mbox{ otherwise } \end{array}\right.
\end{equation}
Recent discussions on the choice, use, and interpretation of the different weighting schemes can be found in \cite{warrens:13} and  \cite{kvalseth:18}. On one side, the main reasons in favor of the quadratic and linear weights are essentially of theoretical nature. In fact, the quadratic weights lead to the interpretation of the weighted kappa as the Intraclass Correlation Coefficient, see \cite{schuster:04}. on the other side, the linear weights allow us to define the weighted kappa as a weighted average of kappas for the $2 \times 2$ tables obtained by collapsing adjacent categories, see \cite{vanbelle|albert:09}. However, undesirable behaviours of the weighted kappa for some data set can be observed under both choices of the weights, and thus the interpretation of the value of kappa is not easy in general. We will come back to this issue later in the paper, when we will use Markov moves to find the maximum agreement. Another interesting interpretation of the linear and quadratic weights is discussed in \cite{li:16}, where matrix $W$ is decomposed into a sum of suitable rank-one matrices.

In order to illustrate our theory, we also consider a square-root version of the weights, namely:
\begin{equation}\label{weight-sqrt}
u_{ij} =  \frac {\sqrt{|i-j|}}{\sqrt{k-1}}
\end{equation}
As a preliminary remark, notice that the linear weights in Eq.~\eqref{weight-linear} and the square-root weights in Eq.~\eqref{weight-squared} define a distance in ${\mathbb R}$, while the quadratic weights in Eq.~\eqref{weight-squared} do not, because the triangular inequality is not satisfied. Usually, functions like the quadratic weights are called dissimilarities. 

In this paper, when the matrix $U$ is a distance matrix, we name the weights as ``distance weights'', and in particular we refer to the weights in Eq.~\eqref{weight-sqrt} as to the sqrt weights. 

Observe that the distance defined by the linear weights is the usual Euclidean distance in ${\mathbb R}$, and it has a special behavior in terms of the triangular inequality. In fact, for $i<j<h$ the triangular inequality becomes an equality: $u_{ih}=u_{ij}+u_{jh}$. We will exploit this property later in the paper.

In this paper we use the notation $\kappa_q$, $\kappa_w$, $\kappa_s$ when quadratic, linear, or sqrt weights are used, while we denote with $\kappa_w$ the kappa with a general weight.

In the multi-rater setting, we consider the ratings of $r$ raters as $r$ random variables $X_1, \ldots, X_r$ on the same set $\{1, \ldots, k\}$, or more generally on $\{x_1, \ldots, x_k\}$. The observed data form a $m^r$ table. We denote with $p_{i_1 \ldots i_r}$ the probability of the cell $(i_1, \ldots, i_r)$, and with $n_{i_1 \ldots i_r}$ the corresponding observed count on a sample of size $N$. Moreover, we denote with $p^{(u)}$ the one-dimensional marginal distribution of $X_u$, and with $p^{(uv)}$ the two-dimensional marginal distribution of the pair $(X_u,X_v)$.

To measure the agreement in the multi-rater setting, it is customary to use the Conger's $\kappa_C$, originally introduced in \cite{conger:80} and the re-analyzed in several papers, see e.g. \cite{vanbelle:19}. The Conger's $\kappa_c$ is based on a pairwise rater agreement analysis. It is defined as:
\begin{equation}\label{eq:kappar}
\kappa_C = \frac {p_o - p_e} {1 - p_e}
\end{equation}
where $p_o$ is the mean proportion of agreement between all $r(r-1)/2$ pairs of raters, and similarly $p_e$ is the mean proportion of expected agreement between all $r(r-1)/2$ pairs of raters under independence. In formulas,
\begin{equation}\label{eq:obsagrr}
p_o = \frac 2 {r(r-1)} \sum_{u,v \in \{1, \ldots, r\},u<v} \sum_{i=1}^k p^{(uv)}_{ii}
\end{equation}
and 
\begin{equation}\label{eq:expagrr}
p_e = \frac 2 {r(r-1)} \sum_{u,v \in \{1, \ldots, r\},u<v} \sum_{i=1}^k p^{(u)}_{i}p^{(v)}_{i} \, .
\end{equation}

Since the Conger's $\kappa_C$ is based on the two-way margins of the $k^r$ table, it is easy to define a weighted version of the Conger's kappa as follows:
\begin{equation}\label{eq:wkappar}
\kappa_{C,w} = \frac {p_{o,w} - p_{e,w}} {1 - p_{e,w}}
\end{equation}
with
\begin{equation}\label{eq:wobsagrr}
p_{o,w} = \frac 2 {r(r-1)} \sum_{u,v \in \{1, \ldots, r\},u<v} \sum_{i,j=1}^k w_{ij}p^{(uv)}_{ij}
\end{equation}
and 
\begin{equation}\label{eq:wexpagrr}
p_{e,w} = \frac 2 {r(r-1)} \sum_{u,v \in \{1, \ldots, r\},u<v} \sum_{i,j=1}^k w_{ij}p^{(u)}_{i}p^{(v)}_{j}
\end{equation}
In the definition above, the weights are the same for all pairs $u,v$ of raters, but the definition can be easily extended to the case of different weights on different two-way margins. Also notice that the Conger's $\kappa_C$ can be defined in the general case of $g$-wise agreement, as in the original paper \cite{conger:80}. This is done by taking the $\kappa_{C}$ unchanged in Eq.~\eqref{eq:kappar}, and computing the observed agreement and the expected agreement in Equations \eqref{eq:obsagrr} and \eqref{eq:expagrr} on the $g$-way marginal tables instead of the two-way tables. However, when the weighted version $\kappa_{C,w}$ in Equations \eqref{eq:wkappar}, \eqref{eq:wobsagrr}, and \eqref{eq:wexpagrr} is considered, the pairwise agreement is the most reasonable choice, and the extension to the $g$-wise agreement would require new definitions of the weighting schemes.

\section{Markov bases} \label{sect:markov}

In this section we introduce the main tools from Algebraic Statistics needed in our framework. In particular, we define the notion of Markov basis and we compute the relevant Markov bases for the rater agreement problems.

Let $n$ be an observed contingency table, possibly multi-way. 
An integer-valued statistic is a function $T:{\mathbb N}^{k^r} \longrightarrow {\mathbb N}^s$. Since we need to compute the maximum agreement with fixed marginal distributions we are particularly interested in the function 
\begin{equation}\label{eq:marg2}
T : n \longmapsto ((n_{i+})_{i=1,\ldots,k},(n_{+j})_{j=1,\ldots,k})
\end{equation}
in the two-way case and
\begin{equation}\label{eq:margr}
T : n \longmapsto ((n^{(1)}_i)_{i=1,\ldots,k}, \ldots , (n^{(r)}_i)_{j=1,\ldots,k})
\end{equation}
in the general multi-rater case, where $n^{(s)}_i$ is the $i$-th entry of the marginal distribution of the $s$-th rater.

\begin{definition} \label{fiber-def}
Given a statistic $T$, the fiber (or reference set) of a contingency table $n$ is the set
\begin{equation}\label{fiber}
{\mathcal{F}}_T(n) = \{ n' \in {\mathbb N}^{m^r} \ | \ T(n') = T(n) \} \, .
\end{equation}
\end{definition}

\begin{definition}
A Markov move for the statistic $T$ is an integer-valued table $m$ such that $T(m)=0$. 
\end{definition}

\begin{definition}
A Markov basis for the fiber ${\mathcal F}_T(n)$ of a contingency table $n$ is a set of Markov moves 
\[
{\mathcal M}_{n,T} = \{m^{(1)}, \ldots, m^{(\ell)} \} 
\]
such that for each pair of tables $n',n'' \in {\mathcal F}_T(n)$ there exists a sequence of moves $(m^{(i_1)}, \ldots, m^{(i_Q)})$ such that
\begin{enumerate}
    \item $n'' = n' + \sum_{j=1}^Q m^{(i_j)}$
    
    \item $n' + \sum_{j=1}^q m^{(i_j)} \geq 0$ for all $q=1, \ldots ,Q$.
\end{enumerate}
\end{definition}
In words, a Markov basis is a set of moves which makes the fiber ${\mathcal F}_T(n)$ connected and all intermediate steps are non-negative. In Algebraic Statistics, this is the main tool to define a Metropolis-like MCMC algorithm for doing exact inference for contingency tables. For a comprehensive introduction to Markov bases and their use in Statistics, the reader can refer to the books \cite{sullivant:18} and \cite{aoki|etal:12}.

Since we are particularly interested in the computation of the maximum agreement given the marginal distributions, we need the Markov bases for the statistic $T$ in Equations \eqref{eq:marg2} and \eqref{eq:margr}

Following \cite{diaconis|sturmfels:98}, in the general case the computation of a Markov basis needs symbolic computation and is actually not feasible for large-sized tables. However, the Markov bases for the fibers considered in this paper can be theoretically characterized and therefore no symbolic computation is involved.  For an overview on the computation of Markov bases through symbolic software, the underlying computational problems, and the actual limitations for large tables, the reader can refer to \cite{aoki|etal:12}.

As a first step, we recall a result from \cite{diaconis|sturmfels:98} about the Markov basis for two-way tables with fixed margins. 

\begin{definition} \label{def:bm2}
Let $i,i'$ be two distinct row indices and $j,j'$ be two distinct column indices. A basic move is a move $m$ such that
\[
m_{ij}=m_{i'j'}=+1, \qquad m_{ij'}=m_{i'j}=-1
\]
and is $0$ otherwise.
\end{definition}

Some examples of basic moves in the case of $4$ categories are given in Fig.~\ref{fig:mosse2}. Such moves have different behavior in terms of agreement. We will discuss all these types of moves in the next section.

\begin{figure}
    \centering
    \begin{tabular}{ccc}
$\begin{pmatrix}
+1 & 0 & -1 & 0  \\
0 & 0 & 0 & 0 \\
-1 & 0 & +1 & 0 \\
0 & 0 & 0 & 0 
\end{pmatrix}
$
&
$\qquad$
&
$\begin{pmatrix}
0 & -1 & 0 & +1  \\
0 & +1 & 0 & -1 \\
0 & 0 & 0 & 0  \\
0 & 0 & 0 & 0 
\end{pmatrix}
$ \vspace{5pt} \\
(a) & $\qquad$ & (b) \\
$\qquad$ & $\qquad$ & $\qquad$ \\
$\begin{pmatrix}
0 & 0 & 0 & 0  \\
0 & +1 & 0 & -1 \\
0 & -1 & 0  & +1 \\
0 & 0 & 0 & 0 
\end{pmatrix}
$
&
$\qquad$
&
$\begin{pmatrix}
0 & 0 & -1 & +1  \\
0 & 0 & +1 & -1 \\
0 & 0 & 0 & 0  \\
0 & 0 & 0 & 0 
\end{pmatrix} 
$ \vspace{5pt} \\
(c) & $\qquad$ & (d) \\

\end{tabular}
 \caption{Four basic moves for the two-rater problem. (a): two non-zero elements on the diagonal; (b): one non-zero element on the diagonal, the move lies on the upper triangle; (c): one non-zero element on the diagonal, the move lies on both the upper and the lower triangle; (d): no non-zero elements on the diagonal.}
    \label{fig:mosse2}
\end{figure}

\begin{proposition}
The set of basic moves in Def.~\ref{def:bm2} is a Markov basis for the fiber in Eq.~\eqref{fiber} for the two-rater problem.
\end{proposition}

The basic moves in the multi-rater setting are defined by extending the previous definition to more than two dimensions. Informally, one takes two $+1$'s in two cells with at least two distinct coordinates, and then arranges the $-1$'s in order to have the correct projections in all two-way marginals. More formally, we can state the following definition.

\begin{definition} \label{def:bmr}
A basic move $m$ for the multi-rater problem is a table with $4$ entries different from $0$:
\begin{itemize}
    \item $m$ is equal to $+1$ in $(i_1, \ldots i_r)$ and in $(i'_1,\ldots, i'_r)$ with at least two different indices. Without loss of generality, suppose that the distinct indices are $i_1, \ldots, i_q$, $q \ge 2$;
    
    \item $m$ is equal to $-1$ in $(j_1,\ldots, j_q,i_{q+1},\ldots, i_r)$ and in  $(j'_1,\ldots, j'_q,i_{q+1},\ldots, i_r)$ with
    \begin{itemize}
        \item[(i)] $j_s=i_s$, $j'_s=i'_s$ for $s \in {\mathcal S}$
        
        \item[(ii)] $j_s=i'_s$, $j'_s=i_s$ for $s \notin {\mathcal S}$
    \end{itemize}
    where ${\mathcal S}$ is a non-empty subset of $\{1, \ldots, q\}$.
\end{itemize}
\end{definition}

It is easy to see that this definition reduces to Def.~\ref{def:bm2} when $r=2$. Two examples of basic moves in the $3^3$ case are illustrated in Fig.~\ref{fig:mosser}.

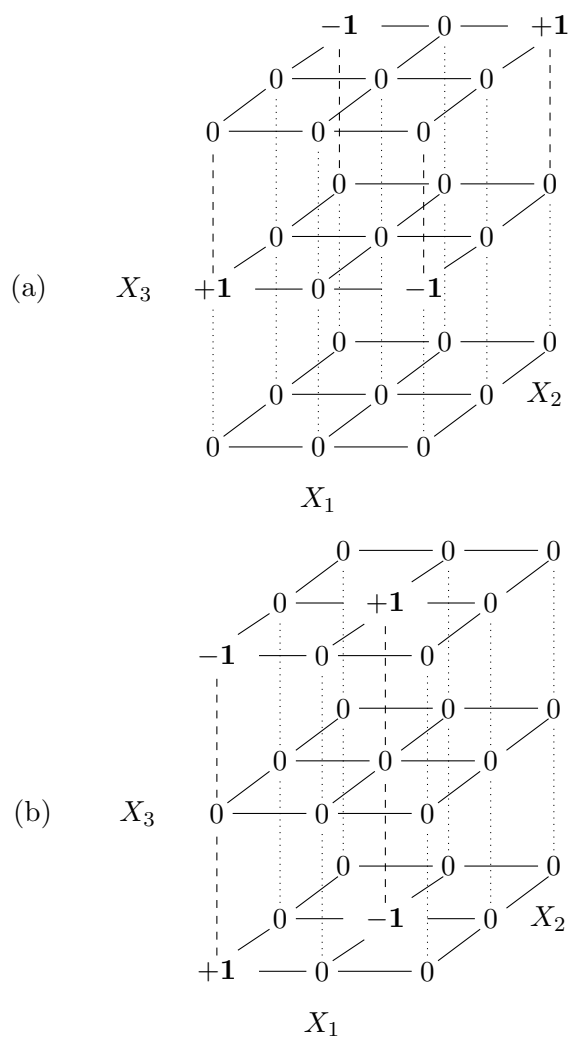
\begin{figure}
    \begin{center}
    \begin{tabular}{c}
\begin{tikzpicture}[scale=0.7]
\draw (3,0) node[scale=1] {$X_1$};
\draw (-0.5,4) node[scale=1] {$X_3$};
\draw (7.3,2) node[scale=1] {$X_2$};

\draw (-2.5,4) node[scale=1] {(a)};

\draw (1,1) node[scale=1] {$0$};
\draw (3,1) node[scale=1] {$0$};
\draw (5,1) node[scale=1] {$0$};
\draw (2.2,2) node[scale=1] {$0$};
\draw (4.2,2) node[scale=1] {$0$};
\draw (6.2,2) node[scale=1] {$0$};
\draw (3.4,3) node[scale=1] {$0$};
\draw (5.4,3) node[scale=1] {$0$};
\draw (7.4,3) node[scale=1] {$0$};


\draw (1,4) node[scale=1] {$\mathbf{+1}$};
\draw (3,4) node[scale=1] {$0$};
\draw (5,4) node[scale=1] {$\mathbf{-1}$};
\draw (2.2,5) node[scale=1] {$0$};
\draw (4.2,5) node[scale=1] {$0$};
\draw (6.2,5) node[scale=1] {$0$};
\draw (3.4,6) node[scale=1] {$0$};
\draw (5.4,6) node[scale=1] {$0$};
\draw (7.4,6) node[scale=1] {$0$};

\draw (1,7) node[scale=1] {$0$};
\draw (3,7) node[scale=1] {$0$};
\draw (5,7) node[scale=1] {$0$};
\draw (2.2,8) node[scale=1] {$0$};
\draw (4.2,8) node[scale=1] {$0$};
\draw (6.2,8) node[scale=1] {$0$};
\draw (3.4,9) node[scale=1] {$\mathbf{-1}$};
\draw (5.4,9) node[scale=1] {$0$};
\draw (7.4,9) node[scale=1] {$\mathbf{+1}$};


\draw (1.3,1) -- (2.7,1);
\draw (3.3,1) -- (4.7,1);
\draw (2.5,2) -- (3.9,2);
\draw (4.5,2) -- (5.9,2);
\draw (3.7,3) -- (5.1,3);
\draw (5.7,3) -- (7.1,3);

\draw (1.2,1.2) -- (2,1.8);
\draw (2.5,2.2) -- (3.3,2.8);
\draw (3.2,1.2) -- (4,1.8);
\draw (4.5,2.2) -- (5.3,2.8);
\draw (5.2,1.2) -- (6,1.8);
\draw (6.5,2.2) -- (7.3,2.8);


\draw (1.8,4) -- (2.7,4);
\draw (3.3,4) -- (4.2,4);
\draw (2.5,5) -- (3.9,5);
\draw (4.5,5) -- (5.9,5);
\draw (3.7,6) -- (5.1,6);
\draw (5.7,6) -- (7.1,6);

\draw (1.4,4.4) -- (2,4.8);
\draw (2.5,5.2) -- (3.3,5.8);
\draw (3.2,4.2) -- (4,4.8);
\draw (4.5,5.2) -- (5.3,5.8);
\draw (5.4,4.4) -- (6,4.8);
\draw (6.5,5.2) -- (7.3,5.8);


\draw (1.3,7) -- (2.7,7);
\draw (3.3,7) -- (4.7,7);
\draw (2.5,8) -- (3.9,8);
\draw (4.5,8) -- (5.9,8);
\draw (4.2,9) -- (5.1,9);
\draw (5.7,9) -- (6.6,9);

\draw (1.2,7.2) -- (2,7.8);
\draw (2.5,8.2) -- (3.1,8.6);
\draw (3.2,7.2) -- (4,7.8);
\draw (4.5,8.2) -- (5.3,8.8);
\draw (5.2,7.2) -- (6,7.8);
\draw (6.5,8.2) -- (7.1,8.6);

\draw[dotted] (1,1.3) -- (1,3.7);
\draw[dotted] (3,1.3) -- (3,3.7);
\draw[dotted] (5,1.3) -- (5,3.7);

\draw[dotted] (2.2,2.3) -- (2.2,4.7);
\draw[dotted] (4.2,2.3) -- (4.2,4.7);
\draw[dotted] (6.2,2.3) -- (6.2,4.7);

\draw[dotted] (3.4,3.3) -- (3.4,5.7);
\draw[dotted] (5.4,3.3) -- (5.4,5.7);
\draw[dotted] (7.4,3.3) -- (7.4,5.7);

\draw[dashed] (1,4.3) -- (1,6.7);
\draw[dotted] (3,4.3) -- (3,6.7);
\draw[dashed] (5,4.3) -- (5,6.7);

\draw[dotted] (2.2,5.3) -- (2.2,7.7);
\draw[dotted] (4.2,5.3) -- (4.2,7.7);
\draw[dotted] (6.2,5.3) -- (6.2,7.7);

\draw[dashed] (3.4,6.3) -- (3.4,8.7);
\draw[dotted] (5.4,6.3) -- (5.4,8.7);
\draw[dashed] (7.4,6.3) -- (7.4,8.7);
\end{tikzpicture} \\

\begin{tikzpicture}[scale=0.7]
\draw (3,0) node[scale=1] {$X_1$};
\draw (-0.5,4) node[scale=1] {$X_3$};
\draw (7.3,2) node[scale=1] {$X_2$};

\draw (-2.5,4) node[scale=1] {(b)};

\draw (1,1) node[scale=1] {$\mathbf{+1}$};
\draw (3,1) node[scale=1] {$0$};
\draw (5,1) node[scale=1] {$0$};
\draw (2.2,2) node[scale=1] {$0$};
\draw (4.2,2) node[scale=1] {$\mathbf{-1}$};
\draw (6.2,2) node[scale=1] {$0$};
\draw (3.4,3) node[scale=1] {$0$};
\draw (5.4,3) node[scale=1] {$0$};
\draw (7.4,3) node[scale=1] {$0$};


\draw (1,4) node[scale=1] {$0$};
\draw (3,4) node[scale=1] {$0$};
\draw (5,4) node[scale=1] {$0$};
\draw (2.2,5) node[scale=1] {$0$};
\draw (4.2,5) node[scale=1] {$0$};
\draw (6.2,5) node[scale=1] {$0$};
\draw (3.4,6) node[scale=1] {$0$};
\draw (5.4,6) node[scale=1] {$0$};
\draw (7.4,6) node[scale=1] {$0$};

\draw (1,7) node[scale=1] {$\mathbf{-1}$};
\draw (3,7) node[scale=1] {$0$};
\draw (5,7) node[scale=1] {$0$};
\draw (2.2,8) node[scale=1] {$0$};
\draw (4.2,8) node[scale=1] {$\mathbf{+1}$};
\draw (6.2,8) node[scale=1] {$0$};
\draw (3.4,9) node[scale=1] {$0$};
\draw (5.4,9) node[scale=1] {$0$};
\draw (7.4,9) node[scale=1] {$0$};


\draw (1.8,1) -- (2.7,1);
\draw (3.3,1) -- (4.7,1);
\draw (2.5,2) -- (3.4,2);
\draw (5.0,2) -- (5.9,2);
\draw (3.7,3) -- (5.1,3);
\draw (5.7,3) -- (7.1,3);

\draw (1.4,1.4) -- (2,1.8);
\draw (2.5,2.2) -- (3.3,2.8);
\draw (3.2,1.2) -- (3.8,1.6);
\draw (4.7,2.4) -- (5.3,2.8);
\draw (5.2,1.2) -- (6,1.8);
\draw (6.5,2.2) -- (7.3,2.8);


\draw (1.3,4) -- (2.7,4);
\draw (3.3,4) -- (4.7,4);
\draw (2.5,5) -- (3.9,5);
\draw (4.5,5) -- (5.9,5);
\draw (3.7,6) -- (5.1,6);
\draw (5.7,6) -- (7.1,6);

\draw (1.2,4.2) -- (2,4.8);
\draw (2.5,5.2) -- (3.3,5.8);
\draw (3.2,4.2) -- (4,4.8);
\draw (4.5,5.2) -- (5.3,5.8);
\draw (5.2,4.2) -- (6,4.8);
\draw (6.5,5.2) -- (7.3,5.8);


\draw (1.8,7) -- (2.7,7);
\draw (3.3,7) -- (4.7,7);
\draw (2.5,8) -- (3.4,8);
\draw (5.0,8) -- (5.9,8);
\draw (3.7,9) -- (5.1,9);
\draw (5.7,9) -- (7.1,9);

\draw (1.4,7.4) -- (2,7.8);
\draw (2.5,8.2) -- (3.3,8.8);
\draw (3.2,7.2) -- (3.8,7.6);
\draw (4.7,8.4) -- (5.3,8.8);
\draw (5.2,7.2) -- (6,7.8);
\draw (6.5,8.2) -- (7.3,8.8);

\draw[dashed] (1,1.3) -- (1,3.7);
\draw[dotted] (3,1.3) -- (3,3.7);
\draw[dotted] (5,1.3) -- (5,3.7);

\draw[dotted] (2.2,2.3) -- (2.2,4.7);
\draw[dashed] (4.2,2.3) -- (4.2,4.7);
\draw[dotted] (6.2,2.3) -- (6.2,4.7);

\draw[dotted] (3.4,3.3) -- (3.4,5.7);
\draw[dotted] (5.4,3.3) -- (5.4,5.7);
\draw[dotted] (7.4,3.3) -- (7.4,5.7);

\draw[dashed] (1,4.3) -- (1,6.7);
\draw[dotted] (3,4.3) -- (3,6.7);
\draw[dotted] (5,4.3) -- (5,6.7);

\draw[dotted] (2.2,5.3) -- (2.2,7.7);
\draw[dashed] (4.2,5.3) -- (4.2,7.7);
\draw[dotted] (6.2,5.3) -- (6.2,7.7);

\draw[dotted] (3.4,6.3) -- (3.4,8.7);
\draw[dotted] (5.4,6.3) -- (5.4,8.7);
\draw[dotted] (7.4,6.3) -- (7.4,8.7);
\end{tikzpicture} \\
\end{tabular}

 \caption{Two basic moves for the three-rater problem. A move of type (a) and a move of type (b) from Prop.~\ref{prop:mosser}.}
    \label{fig:mosser}
    \end{center}
\end{figure}

The fact that basic moves are enough to connect the fiber in Eq. \eqref{fiber} can be derived from the theory of toric fiber products to be found in \cite{sullivant:07}. This allows us to avoid symbolic computations and to make available the relevant Markov bases also for large-sized tables.

\begin{proposition} \label{prop:mosser}
The set of basic moves in Def.~\ref{def:bmr} is a Markov basis for the fiber in Eq.~\eqref{fiber} when $r>2$.
\end{proposition}
\begin{proof}
First, note that all basic moves in Def.~\ref{def:bmr} are in the kernel of the marginalization map $T$.

Since for $r=2$ the basic moves in Def.~\ref{def:bmr} coincide with the basic moves in \ref{def:bm2}, the result is true when $r=2$. We proceed by induction on $r$. Let us suppose that the result holds for $r-1$ raters, and we prove it for $r$ raters. 

We apply Theorem 13 in \cite{sullivant:07}. Given a Markov basis ${\mathcal M}_{r-1}$ for the problem with $(r-1)$ raters, a Markov basis ${\mathcal M}_{r}$ for the problem with $r$ raters is the union of the following sets of moves
\begin{itemize}
    \item[(a)] split each move of ${\mathcal M}_{r-1}$ by putting one $+1$ and one $-1$ at a given level $h$ of $X_r$ ($h=1, \ldots , k$) and the other $+1$ and $-1$ at a level $h'$ of $X_r$  ($h'=h, \ldots, k$);
    
    \item[(b)] for any two distinct cells $(i_1, \ldots , i_{r-1})$ and $(i'_1, \ldots , i'_{r-1})$ on the $(r-1)$-dimensional table, and for any two distinct levels $h,h'$ of $X_r$, take the move with $+1$ in $(i_1, \ldots , i_{r-1},h)$ and in $(i'_1, \ldots , i'_{r-1},h')$ and with $-1$ in $(i_1, \ldots , i_{r-1},h')$ and in $(i'_1, \ldots , i'_{r-1},h)$.
\end{itemize}
Since all the moves defined in items $(a)$ and $(b)$ above are basic moves, the result is proved.
\end{proof}

As noticed in the Introduction, Markov bases in Algebraic Statistics are usually  defined in Algebraic Statistics in order to perform exact tests with a Metropolis-Hastings algorithm, and therefore to generate all contingency tables with the same value of the sufficient statistics as the observed table. Here, we simply use Markov bases to compute all the tables with fixed margins.

\section{The effect of the Markov moves on the kappa indices} \label{sect:interpret}

In this section, we use the basic moves of the Markov bases in order to better understand the meaning of the weighted kappa. The basic idea is to apply the definition of Markov basis to analyze the rater agreement in connection with the weighting schemes. In fact, the configuration of maximum agreement can be reached with a finite number of Markov moves, starting from the observed table. Therefore, we analyze how the rater agreement changes when a Markov move is applied.

Since we are especially interested in the analysis of rater agreement with fixed margins, the quantity of interest is the observed agreement
\begin{equation}\label{eq:obsagr2}
A_{o,w}(n) = \frac 1 N \sum_{i,j=1}^k w_{ij}n_{ij}
\end{equation}
in the two-rater setting, and
\begin{equation}\label{eq:obsagrr_v2}
A_{o,w}(n) = \frac 2 {r(r-1)} \sum_{u,v \in \{1, \ldots, r\},u<v} \frac 1 N \sum_{i,j=1}^k w_{ij}n^{(uv)}_{ij}
\end{equation}
in the multi-rater setting.

Let us start with some results in the two-rater setting.

\begin{lemma} \label{first-res}
Let $n$ be an observed agreement table, let $i \ne j$ be two indices, and let $m$ be the basic move with 
\[
m_{ii}=m_{jj}=+1, \qquad m_{ij}=m_{ji}=-1 \, .
\]
If $n_{ij}>0$ and $n_{ji}>0$, then
\begin{equation} \label{eq:symmove}
  A_{o,w}(n+m) \ge A_{o,w}(n)  \, .
\end{equation}
\end{lemma}
\begin{proof}
From Eq.~\eqref{eq:obsagr2} and using the disagreement weights, we get
\[
A_{o,w}(n+m) - A_{o,w}(n) = A_{o,w}(m) = - \frac 1 N \sum_{i,j=1}^k u_{ij}m_{ij} = \frac 2 N u_{ij} \ge 0 \, . 
\]
\end{proof} 

Lemma \ref{first-res} is valid for all weighting schemes and tells us that if there are positive counts in symmetric cells, then it is always possible to construct an observed table with higher observed agreement by applying a simple move. This is quite intuitive, since Eq.~\eqref{eq:symmove} roughly says that moving counts on the diagonal increases the observed agreement.

Nonetheless, apart from the symmetric basic moves as displayed in  Fig.~\ref{fig:mosse2} (a), for the other types of basic moves there is not a common behavior in terms of observed agreement. 
Remember that, earlier in the paper, we have noticed that some weighting schemes defines a distance on the ground set $\{x_1, \ldots, x_k\}$ while other schemes do not. The following proposition states a partial result when only one cell of the diagonal is involved in the basic move.

\begin{proposition} \label{second-res}
Let $n$ be an observed agreement table, let $i < j < h$ be three indices, and let $m$ be the basic move with 
\[
m_{ih}=m_{jj}=+1, \qquad m_{ij}=m_{jh}=-1 \, .
\]. 
If $n_{ij}>0$ and $n_{jh}>0$ and a distance weighting scheme is used, then
\[
  A_{o,w}(n+m) \ge A_{o,w}(n) \, .  
\]
The same holds if $i>j>h$.
\end{proposition}
\begin{proof}
Let us consider the case $i < j < h$. (The other case has a similar proof.) 

From Eq.~\eqref{eq:obsagr2} and using the disagreement weights, we get
\[
A_{o,w}(n+m) - A_{o,w}(n) = A_{o,w}(m) = - \frac 1 N \sum_{i,j=1}^k u_{ij}m_{ij} =
\]
\[
= \frac 1 N \left( u_{ij} + u_{jh}- u_{ih} \right) \ge 0 
\]
by virtue of the triangular inequality.
\end{proof}

In Prop.~\ref{second-res} the move $m$ has one non-zero element on the diagonal. In the case $i < j < h$ the move lies in the upper triangle of the table, while in the case $i>j>h$ lies in the lower one.

Some remarks are now in order. First, note that in Prop.~\ref{second-res} the assumption of distance weights is essential. For weighting schemes derived or not derived from a distance we observe opposite behaviors of the weighted kappa. For instance, let us consider the observed table below:
\[
n=\begin{pmatrix}
4 & 0 & 0 & 0  \\
0 & 4 & 1 & 0 \\
0 & 0 & 4 & 1  \\
0 & 0 & 0 & 4 
\end{pmatrix} \, .
\]
We can apply the move
\[
m=\begin{pmatrix}
0 & 0 & 0 & 0  \\
0 & 0 & -1 & +1 \\
0 & 0 & +1 & -1  \\
0 & 0 & 0 & 0 
\end{pmatrix}
\]
and we obtain 
\[
n'=n+m=\begin{pmatrix}
4 & 0 & 0 & 0  \\
0 & 4 & 0 & 1 \\
0 & 0 & 5 & 0  \\
0 & 0 & 0 & 4 
\end{pmatrix} \, .
\]
Comparing the value of $\kappa_w$ of $n$ and $n'$ we note that:
\begin{itemize}
    \item With a distance weight we have $\kappa_w(n')>\kappa_w(n)$ by virtue of Prop.~\ref{second-res};
    \item With quadratic weights we have $\kappa_q(n')<\kappa_q(n)$;
\item With linear weights we get $\kappa_l(n')=\kappa_l(n)$.
\end{itemize}

While with distance weights the maximum agreement is achieved by maximizing the counts in the diagonal cells, with quadratic weights a certain amount of moderate disagreement is preferred to a small amount of strong disagreement. 

Moreover, from the example above, we observe there is a special behavior of the linear weights, because some Markov moves do not change the value of the weighted kappa. This affects also the problem of finding the configuration with maximum agreement, since in general such a configuration is not unique. We state below a result for linear weights, and we will discuss this issue in the next section from the point of view of computations.

\begin{proposition} \label{prop:linear}
Let us consider four indices $i_1<i_2\le j_1 <j_2$ or $j_1 < j_2 \le i_1 < i_2$, and take the basic move $m$ with $m_{i_1j_1}=m_{i_2j_2}=+1$ and $m_{i_1j_2}=m_{i_2j_1}=-1$. If $n$ is a table with $n_{i_1j_2}>0$ and $n_{i_2j_1}>0$. Using the linear weights we get
\[
\kappa_l(n) = \kappa_l(n+m) \, .
\]
\end{proposition}
\begin{proof}
As in the previous proposition, let us consider only the case $i_1<i_2\le j_1 <j_2$.

Notice that the conditions $n_{i_1j_2}>0$ and $n_{i_2j_1}>0$ are needed in order to have a non-negative table $n'=n+m$. Since $n$ and $n'$ have the same margins, it is enough to compare the observed agreement. From Eq.~\eqref{eq:obsagrr_v2} we get:
\[
A_{o,w}(n')-A_{o,w}(n) = A_{o,w}(m) = - \frac 1 N \sum_{i,j=1}^k, u_{ij}m_{ij} = 
\]
\[
= -\frac 1 N \cdot \frac {(j_1-i_i)+(j_2-i_2)-(j_2-i_1)-(j_1-i_2)} {k-1} = 0 \, .
\]
\end{proof}

The condition $i_1<i_2\le j_1 <j_2$ means that we apply a move on one side of the table w.r.t. the diagonal, and one non-zero element of the move is on the diagonal when $i_2=j_1$. Under such a condition, the move does not affect the value of the weighted kappa. 

In view of Prop.~\ref{prop:linear}, the uniqueness of the table with a given value of $\kappa_w$ is not guaranteed under any of weighting schemes, but this issue is especially relevant for the linear weights. To illustrate this, let us consider the table (with synthetic data) in Fig.~\ref{fig:eslin} (a). By direct enumeration of the $644,850$ tables of the fiber, one finds $1,527$ tables with the same margins and with the same value of the weighted kappa with linear weights as the observed table, i.e., $\kappa_l=0.5023$. Among those tables, the weighted kappa with quadratic weights ranges from $\kappa_q=0.3774$ to $\kappa_q=0.7406$. The minimum is achieved in 3 tables, one of which is in Fig.~\ref{fig:eslin} (b), while the maximum is achieved in 3 tables, one of which is in Fig.~\ref{fig:eslin} (c).

\begin{figure}
    \centering
    \begin{tabular}{ccc}
\begin{tabular}{cc|cccc}
& & \multicolumn{4}{c}{$Y$} \\ 
 & & 1 & 2 & 3 & 4 \\ \hline
\multirow{4}{*}{$X$} & 1 & 5 & 3 & 2 & 1 \\
& 2 & 1 & 4 & 3 & 0 \\
& 3 & 0 & 1 & 5 & 1 \\
& 4 & 0 & 1 & 2 & 4 
\end{tabular}
&
\begin{tabular}{cc|cccc}
& & \multicolumn{4}{c}{$Y$} \\ 
 & & 1 & 2 & 3 & 4 \\ \hline
\multirow{4}{*}{$X$} & 1 & 6 & 0 & 1 & 4 \\
& 2 & 0 & 8 & 0 & 0 \\
& 3 & 0 & 0 & 7 & 0 \\
& 4 & 0 & 1 & 4 & 2 
\end{tabular}
&
\begin{tabular}{cc|cccc}
& & \multicolumn{4}{c}{$Y$} \\ 
 & & 1 & 2 & 3 & 4 \\ \hline
\multirow{4}{*}{$X$} & 1 & 6 & 5 & 0 & 0 \\
& 2 & 0 & 3 & 5 & 0 \\
& 3 & 0 & 1 & 2 & 4 \\
& 4 & 0 & 0 & 5 & 2 
\end{tabular} \vspace{5pt}\\
(a) & (b) & (c) \\
\end{tabular}
 \caption{A synthetic observed table (a) and two tables with the same margins and with the same weighted kappa under linear weights (b,c).}
    \label{fig:eslin}
\end{figure}

Let us now turn to the multi-rater setting. From Eq.~\eqref{eq:wkappar}, \eqref{eq:obsagrr}, \eqref{eq:expagrr} it is easy to argue that the effect of a basic move on the value of the Conger's $\kappa_C$ is yielded by the two-way margins of the move. Each two-way projection applies to a two-way margin and gives its own contribution in the sum in Eq.~\eqref{eq:obsagrr}.

The following proposition collects the properties of the two-way margins of a basic move, and its proof is immediate.

\begin{proposition}
Let $m$ be a basic move in the multi-rater case with $r$ raters. Suppose that $m$ is equal to $+1$ in $(i_1, \ldots ,i_r)$ and in $(i'_1,\ldots, i'_r)$ and is equal to $-1$ in $(j_1, \ldots,j_r)$ and in $(j'_1, \ldots, j'_r)$. The projection of $m$ on the pair $(U,V)$ is:
\begin{itemize}
    \item a basic move for the two-way problem, if the four pairs $(i_u, i_v)$, $(i'_u,i'_v)$, $(j_u,j_v)$, $(j'_u,j'_v)$ are all distinct;
    
    \item a null move, otherwise.
\end{itemize}
\end{proposition}

Note that, following the definition of basic move in Def.~\ref{def:bmr}, it is easy to see that a multi-rater basic move $m$ always yields at least one basic move on some two-dimensional margin.

For example, both the moves for the three-rater problem displayed in Fig.~\ref{fig:mosser} produce a basic move on two two-way margins and a null move on one margin.

In general, the analysis of the effect of the basic moves on the Conger's $\kappa_C$ is more difficult than in the two-rater case. Nevertheless, we can state the following lemma, which generalizes Lemma \ref{first-res}.

\begin{lemma} \label{first-res-multi}
Let $i \ne j$ be two indices, and let $m$ be a basic move with \[
m_{i\ldots i}=m_{j\ldots j}=+1 \, .
\]
If $n+m$ is non-negative, then
\begin{equation}
  A_{o,w}(n+m) \ge A_{o,w}(n)  \, .
\end{equation}
\end{lemma}
\begin{proof}
It is enough to observe that all two-way margins of $m$ are either a null move or a basic move satisfying the hypothesis of Lemma \ref{first-res}.
\end{proof}

\begin{figure}
    \centering
    \begin{tabular}{ccc}
\begin{tabular}{cc|ccc}
& & \multicolumn{3}{c}{$X_3$} \\ 
 & & 1 & 2 & 3 \\ \hline
\multirow{3}{*}{$X_2$} & 1 & 2 & 1 & 0 \\
& 2 & 0 & 1 & 0 \\
& 3 & 0 & 0 & 1 \\
\end{tabular}
&
\begin{tabular}{cc|ccc}
& & \multicolumn{3}{c}{$X_3$} \\ 
 & & 1 & 2 & 3 \\ \hline
\multirow{3}{*}{$X_2$} & 1 & 0 & 1 & 0 \\
& 2 & 1 & 3 & 1 \\
& 3 & 0 & 0 & 0 \\
\end{tabular} 
&
\begin{tabular}{cc|ccc}
& & \multicolumn{3}{c}{$X_3$} \\ 
 & & 1 & 2 & 3 \\ \hline
\multirow{3}{*}{$X_2$} & 1 & 0 & 1 & 0 \\
& 2 & 0 & 1 & 0 \\
& 3 & 0 & 0 & 3 \\
\end{tabular} \vspace{5pt}\\
$X_1=1$ & $X_1=2$ & $X_1=3$ \\
\end{tabular}
 \caption{An observed table with 3 raters and 3 levels.}
    \label{fig:tab3x3}
\end{figure}

Also the multi-rater case the issue of non-uniqueness is especially relevant under linear wights. For instance, let us consider the three way table in Fig.~\ref{fig:tab3x3}. Although the table is rather sparse (the sample size is $16$ in a contingency table with $27$ cells), there are $2,324$ tables with the same value of $\kappa_{C,l}=0.4872$. Under quadratic weights, such tables yield values of $\kappa_{C,q}$ ranging from $0.3364$ to $0.6313$.

\section{Simulated annealing for maximum agreement} \label{sect:algo}

In this section, we show how to use a simulated annealing algorithm to determine the maximum value of the weighted kappa with fixed marginal distributions and to find a table where the maximum is actually reached. The Markov bases introduced in Section \ref{sect:markov} are used in the algorithm to define the neighbours of the contingency tables and to navigate the fiber of an observed table. 

While the computation of the maximum agreement is simple for the unweighted $\kappa$ in the two-rater setting, the problem is not trivial when the weighted $\kappa_w$ is considered, or we use the Conger's $\kappa_C$ or $\kappa_{C_w}$ for the multi-rater problem. 

The MCMC simulated annealing algorithm starts from the observed table and runs at each step $b$ ($b=1, \ldots, B$) as follows. First, we choose a move $m$ in the relevant Markov basis ${\mathcal M}$ and we define $n'=n+m$; if $n'$ is a  non-negative table, then we move the chain from $n$ to $n'$ with a transition probability depending on two factors. On one hand, the transition probability is equal to 1 if the move causes an increase in the observed agreement, while it is less than one if the move causes a decrease in the observed agreement, and this probability is lower the more the decrease is high. On the other hand, the transition probability decreases with the time. 

In practice, in the first part of the walk the MCMC procedure performs exploration, while in the second part it performs exploitation, because the probability of an actual move toward a  table with smaller observer agreement decreases with the  time. With our notation, the formula for the transition probability is:
\[
\min\left\{\exp((A_{o,w}(n')-A_{o,w}(n'))/\tau_b), 1 \right\} \ ,
\]
where $A_{o,w}$ is the observed agreement in the table $n$, and $\tau_b$ is the temperature at time $b$. 

As a special feature of this algorithm, we have added a final step to apply all possible moves with two $+1$ on the main diagonal, thus exploiting the results in Lemmas \ref{first-res} and \ref{first-res-multi}. 

The pseudo-code of the algorithm is in Fig.~\ref{fig:alg}.

The reader can refer to \cite{suman|kumar:06} for a general introduction to simulated annealing in the discrete case and for a discussion on the computational details of the algorithm, as for instance the choice of the temperature function $\tau_b$. In particular, from our experiments, the choice of the function for the temperature decrease does not affect the performance of the algorithm, and thus we have used a temperature of the form $\tau=\tau_0 \cdot d^b$.

\begin{figure}
\begin{algorithm}[H]
\DontPrintSemicolon
  \KwInput{The observed table $n_{\mathrm{obs}}$}
  \KwOutput{A table with maximum agreement}
  \KwData{Markov basis ${\mathcal M}$; Initial temperature $\tau$; number of MCMC steps $B$}
 initialize $n=n_{\mathrm{obs}}$ \\
  \For{$b$ in $1:B$}
    {
        Choose a basic move $m$ in ${\mathcal M}$ \\
        Define $n'=n+m$ \\
        \If{$n' \ge 0$}
        {
         Define $p_t=\min\left\{\exp((A_{o,w}(n')-A_{o,w}(n))/\tau), 1 \right\}$ \\
         Generate $u \sim {\mathcal U}(0,1)$ \\
         \If{$p_t>u$}
         {
          $n=n'$
         }
        }
        Decrease $\tau$
    }
  \For {each move $m$ with two $+1$ on the diagonal}
    {
      Define $n'=n+m$ \\
      \If {$n' \ge 0$}
      {
        $n=n'$
      }
    }
  \Return $n$
  
\end{algorithm} \caption{Simulated annealing for maximum agreement.} \label{fig:alg}
\end{figure} 

Notice that the non-uniqueness of the configuration is still an issue also when finding the maximum, especially using the linear weights. As an example in the two-rater framework, consider again the observed table in Fig.~\ref{fig:eslin} (a). With linear weights there are $5$ tables which reach the maximum value of $\kappa_l=0.7511$, and among these tables the $\kappa_q$ ranges from $0.7665$ to $0.8703$, the latter being also the maximum with quadratic weights. The maximum with the sqrt weights is $\kappa_s=0.7528$. The three configurations obtained with our algorithm are displayed in Fig.~\ref{fig:max2}. In accordance with the findings in the previous sections, we note that quadratic weights avoid strong disagremeent cells, while sqrt weights fill the diagonal as much as possible. Again, the table with maximum linear weight is not unique, and in fact the three tables in Fig.~\ref{fig:max2} share the same of $\kappa_l$.

\begin{figure}
    \centering
    \begin{tabular}{ccc}
\begin{tabular}{cc|cccc}
& & \multicolumn{4}{c}{$Y$} \\ 
 & & 1 & 2 & 3 & 4 \\ \hline
\multirow{4}{*}{$X$} & 1 & 6 & 5 & 0 & 0 \\
& 2 & 0 & 4 & 4 & 0 \\
& 3 & 0 & 0 & 7 & 0 \\
& 4 & 0 & 0 & 1 & 6 
\end{tabular}
&
\begin{tabular}{cc|cccc}
& & \multicolumn{4}{c}{$Y$} \\ 
 & & 1 & 2 & 3 & 4 \\ \hline
\multirow{4}{*}{$X$} & 1 & 6 & 3 & 2 & 0 \\
& 2 & 0 & 6 & 2 & 0 \\
& 3 & 0 & 0 & 7 & 0 \\
& 4 & 0 & 0 & 1 & 6 
\end{tabular}
&
\begin{tabular}{cc|cccc}
& & \multicolumn{4}{c}{$Y$} \\ 
 & & 1 & 2 & 3 & 4 \\ \hline
\multirow{4}{*}{$X$} & 1 & 6 & 1 & 4 & 0 \\
& 2 & 0 & 8 & 0 & 0 \\
& 3 & 0 & 0 & 7 & 0 \\
& 4 & 0 & 0 & 1 & 6 
\end{tabular} \vspace{5pt}\\
$\kappa_q= 0.8703$ (max) & $\kappa_q=0.8184$ & $\kappa_q= 0.7665$ \\
$\kappa_l= 0.7511 $ & $\kappa_l=0.7511$ (max) & $\kappa_l= 0.7511$ \\
$\kappa_s= 0.6771$ & $\kappa_s=0.7150$ & $\kappa_s= 0.7528$ (max) \\
\end{tabular}
 \caption{Configurations with maximum weighted kappa for the observed table in Fig.~\ref{fig:eslin} with quadratic weights (left), linear weights (center), sqrt weights (right).}
    \label{fig:max2}
\end{figure}

The algorithm converges very fast, at least for small- and medium-sized tables, yielding the maximum value of the weighted kappa and a table where such a maximum is reached in less than 1 second on a standard PC. For large tables the convergence takes long times, and the problem becomes fast unfeasible when the number of cells is large. In fact, on one side large tables are usually sparse, on the other side the relevant Markov basis is large, and at each step the probability of an applicable move is very low. As a consequence, for large tables the number of MCMC steps $B$ must be quite large to ensure convergence. Some experiments are shown through a simulation study in the next section. In our experiments, we have found a fast convergence for tables with up to 300 cells: on a standard PC, the algorithm runs in less than 1 second for tables up to 100 cells, and in less than 10 seconds for tables up to 300 cells.

Note that one can replace the fixed run length $B$ with a stopping rule, and this is the strategy implemented in our simulation study. For instance, in small problems one can stop the algorithm when the algorithm does produce actual moves for $1,000$ consecutive steps. For large tables the stopping rule must take into account also the cardinality of the Markov basis. More details on this point are discussed in the next section.

In general, the use of Algebraic Statistics in the case of large tables is problematic, and the curse of dimensionality is a known issue of this discipline. The definition of new techniques to speed up the convergence of MCMC algorithms within Algebraic statistics is still a current research topic, see for instance \cite{windisch:16}, and only {\it ad hoc} solutions for special problems are currently available.

\section{Simulation study} \label{sect:simst}

In order to show the practical applicability of the algorithm introduced in the previous section, and to study its convergence properties, we have designed and performed a simulation study with several scenarios. For the two-rater case, we have considered three values of the number of levels $k$ ($k=3,5,7$) and two sample sizes ($N=20,100$). Moreover, two types of marginal distributions are considered: a first case with homogeneous uniform marginals, and a second case with non-homogeneous marginals. In the first case, the tables are generated from a multinomial distribution with probabilities given by $\mu \otimes \mu$, with $\mu=(1/k, \ldots, 1/k)$, while in the second case the probability parameter of the multinomial distribution is $\mu \otimes \nu$ with $\mu \ne \nu$. In the non-homogeneous case, the parameters $\mu$ and $\nu$ are chosen to account for the tendency of a rater to choose rating levels higher or lower than those of the other rater. For instance, in the $3 \times 3$ case, we have used $\mu=(2/5,2/5,1/5)$ and $\nu=(1/5,2/5,2/5)$.

Notice that, with this procedure, we obtain observed marginal distributions also when the parameter of the multinomial distribution is of the form $\mu \otimes \mu$, and thus the problem of finding the maximum weighted kappa is not trivial even in these scenarios.

Also a simulation study for the three-rater case is presented, but limited to two numbers of categories $k=3,5$.

The convergence of the algorithm is measured as follows. The algorithm stops when there is a sufficiently large number $c$ of consecutive steps with no change in the observed agreement (and therefore without changes in the the weighted kappa). The number $c$ must take into account the number of moves in the Markov basis. We have defined here $c=\max\{10\cdot \#{\mathcal M}; 1,000\}$. This choice of $c$ is a reasonable trade-off between accuracy and speed.

For each scenario, a sample of $1,000$ tables is generated and the distribution of the stopping time $T$ is approximated through the $1,000$ observed values. The simulation study has been performed using three weighing schemes: quadratic, linear, and sqrt.

The results are displayed in Table \ref{tab:tabr2hm} for the two-rater scenarios and in Table \ref{tab:tabr3hm} for the there-rater scenarios. The mean, the standard deviation, and the $99$th percentile of the convergence time $T$ are reported. In such tables, only the results for tables with homogeneous marginals are considered. Since the results for tables with non-homogeneous marginals are very similar, they are reported as Tables \ref{tab:tabr2nhm} and \ref{tab:tabr3nhm} in the Appendix.

\begin{table}[htbp]
    \centering
    \begin{tabular}{c|c|c|c|c|c}
    Weight & $k$ & $N$ & mean & sd & $q_{0.99}$ \\ \hline
    Quadratic & 3 & 20 & 1,064.8 & 32.0 & 1,164 \\
              &   & 100 & 1,264.5 & 75.2 & 1,481  \\
              & 5 & 20 & 2,568.5 & 277.7 & 3,418  \\
              &   & 100 & 3,943.6 & 637.9 & 5,922 \\
              & 7 & 20 & 11,300.4 & 1,157.9 & 14,914  \\
              &   & 100 & 15,497.8 & 2,211.2 & 21,997 \\ \hline
       Linear & 3 & 20 & 1,049.4 & 25.3 & 1,126  \\
              &   & 100 & 1,170.8 & 47.6 & 1,307  \\
              & 5 & 20 & 2,407.3 & 249.1 & 3,233  \\
              &   & 100 & 3,064.6 & 428.6 & 4,420  \\
              & 7 & 20 & 10,532.5 & 1,042.5 & 13,784  \\
              &   & 100 & 12,358.5 & 1,516.0 & 16,941 \\ \hline
       Sqrt   & 3 & 20 & 1,055.6 & 27.3 & 1,139  \\
              &   & 100 & 1,165.5 & 44.6 & 1,290  \\
              & 5 & 20 & 2,499.6 & 257.1 & 3,298  \\
              &   & 100 & 3,088.3 & 401.9 & 4,318  \\
              & 7 & 20 & 10,955.3 & 1,018.3 & 14,054  \\
              &   & 100 & 12,872.3 & 1,539.2 & 18,235 \\ \hline
    \end{tabular}
    \caption{Two-rater case with homogeneous marginal distributions. Time to convergence (mean, standard deviation and 99th percentile) of the simulated annealing algorithm for different numbers of levels $k$ and sample sizes $N$.}
    \label{tab:tabr2hm}
\end{table}

\begin{table}[htbp]
    \centering
    \begin{tabular}{c|c|c|c|c|c}
    Weight & $k$ & $N$ & mean & sd & $q_{0.99}$ \\ \hline
    Quadratic & 3 & 20 & 4,227.8 & 491.8 & 5,778   \\
              &   & 100 & 5,753.2 & 830.9 & 8,237  \\
              & 5 & 20 & 115,590.4 & 11,611.1 & 153,122 \\
              &   & 100 & 142,858.7 & 17,776.5 & 190,784 \\ \hline
       Linear & 3 & 20 & 4,057.9 & 402.7 & 5,619  \\
              &   & 100 & 5,075.6 & 629.0 & 6,987  \\
              & 5 & 20 & 110,198.5 & 10,633.9 & 146,114   \\
              &   & 100 & 124,430.1 & 13,703.3 & 167,849 \\ \hline
       Sqrt   & 3 & 20 & 4,213.5 & 508.0 & 5,842  \\
              &   & 100 & 5,248.2 & 754.4 & 7,760  \\
              & 5 & 20 & 117,436.7 & 12,774.8 & 159,911 \\
              &   & 100 & 131,752.7 & 15,281.4 & 178,512 \\ \hline
    \end{tabular}
    \caption{Three-rater case with homogeneous marginal distributions. Time to convergence (mean, standard deviation and 99th percentile) of the simulated annealing algorithm for different numbers of levels $k$ and sample sizes $N$.}
    \label{tab:tabr3hm}
\end{table}

From the results, we see that the time to convergence increases with the sample size and with the dimension of the table, and this is particularly relevant in the three-rater case. As discussed in the previous sections, when the number of raters increases, the number of basic moves in the Markov basis grows, and the probability of selecting a non-applicable move becomes high, especially in the case of sparse tables. To overcome this problem, the definition of the stopping time $c$ requires a large number of steps when the Markov basis is large and consequently the execution time increases. For large sparse tables, the algorithm needs special attention in the choice of the numerical parameters and in the optimization of the selection of the moves. A thorough study in this direction is beyond the scopes of the present paper. That is why we do not present the case of $7 \times 7 \times 7$ tables. Finally, with regard to the choice of the weights, we observe that the algorithm is a bit faster with the linear weights.

\section{Concluding remarks} \label{sect:final}

The analysis of the kappa-type indices through basic Markov moves presented in this paper allows us to better understand the effect of the choice of the weights, and in particular shows that the the configuration with maximum kappa strongly depends on the weights, making the normalization of the kappa statistics a non trivial task. We have shown that, when the weights satisfy the triangular inequality, the table with maximum kappa looks quite different from that obtained with quadratic weights, and therefore the use of distance weights should be considered as an option when choosing the weights.

Since the basic moves make connected the fiber of all tables with the same margins, we have implemented an MCMC algorithm to actually find the configuration with maximum kappa with fixed margins in a general framework.

Future works will include the analysis of the maximum agreement when not all raters classify the same set of objects, and the speed up of the MCMC algorithm, especially for large sparse tables. The convergence of MCMC algorithms with Markov bases for large sparse tables is a general problem in Algebraic Statistics, and thus any advance in this direction would represent a notable progress also in other fields of application. Finally, we have shown that the set of all tables with a given value of weighted kappa with linear weights can be rather large. A thorough analysis of such a set can be performed with the use of suitable Markov bases.

\bibliographystyle{alpha}
\bibliography{biblio_FRagr}

\newpage

\section*{Appendix A}

In this appendix the results of the simulation study with non-homogeneous margins are reported. See Sect.~\ref{sect:simst} for the description of the simulation study.

\begin{table}[htbp]
    \centering
    \begin{tabular}{c|c|c|c|c|c}
    Weight & $k$ & $N$ & mean & sd & $q_{0.99}$ \\ \hline
    Quadratic & 3 & 20 & 1,064.7 & 34.0 & 1,172  \\
              &   & 100 & 1,180.8 & 66.7 & 1,381  \\
              & 5 & 20 & 2,572.8 & 302.0 & 3,599  \\
              &   & 100 & 3,859.4 & 698.4 & 6,099 \\
              & 7 & 20 & 11,265.9 & 1,192.6 & 14,957  \\
              &   & 100 & 15,247.9 & 2,191.4 & 22,060 \\ \hline
       Linear & 3 & 20 & 1,039.0 & 24.8 & 1,122  \\
              &   & 100 & 1,106.4 & 35.8 & 1,205  \\
              & 5 & 20 & 2,294.0 & 208.3 & 3,072  \\
              &   & 100 & 2,514.7 & 233.7 & 3,283   \\
              & 7 & 20 & 9,989.9 & 863.7 & 13,430  \\
              &   & 100 & 10,384.7 & 874.3 & 13,437 \\ \hline
         Sqrt & 3 & 20 & 1,058.8 & 33.0 & 1,169  \\
              &   & 100 & 1,222.3 & 84.9 & 1,457  \\
              & 5 & 20 & 2,528.4 & 274.4 & 3,446 \\
              &   & 100 & 3,438.4 & 596.2 & 5,350   \\
              & 7 & 20 & 11,138.2 & 1,173.5 & 15,031  \\
              &   & 100 & 13,858.6 & 1,982.0 & 19,791  \\ \hline
    \end{tabular}
    \caption{Two-rater case with non-homogeneous marginal distributions. Time to convergence (mean, standard deviation and 99th percentile) of the simulated annealing algorithm for different numbers of levels $k$ and sample sizes $N$.}
    \label{tab:tabr2nhm}
\end{table}

\begin{table}[htbp]
    \centering
    \begin{tabular}{c|c|c|c|c|c}
    Weight & $k$ & $N$ & mean & sd & $q_{0.99}$ \\ \hline
    Quadratic & 3 & 20 & 4,210.5 & 496.6 & 5,881  \\
              &   & 100 & 5,839.7 & 868.4 & 8,607 \\
              & 5 & 20 & 115,001.6 & 11,839.8 & 152,551  \\
              &   & 100 & 145,570.5 & 19,149.3 & 210,442  \\ \hline
       Linear & 3 & 20 & 4,085.2 & 466.1 & 5,825  \\
              &   & 100 & 5,395.4 & 775.2 & 7,640 \\
              & 5 & 20 & 110,907.2 & 10,893.4 & 146,344  \\
              &   & 100 & 133,778.4 & 17,015.5 & 187,251   \\ \hline
        Sqrt  & 3 & 20 & 4,201.5 & 519.9 & 5,937 \\
              &   & 100 & 5,408.6 & 791.3 & 7,958  \\
              & 5 & 20 & 117,136.2 & 12,284.0 & 156,774  \\
              &   & 100 & 136,924.9 & 16,926.4 & 187,861 \\ \hline
    \end{tabular}
    \caption{Three-rater case with non-homogeneous marginal distributions. Time to convergence (mean, standard deviation and 99th percentile) of the simulated annealing algorithm for different numbers of levels $k$ and sample sizes $N$.}
    \label{tab:tabr3nhm}
\end{table}

\end{document}